\documentclass[conference]{IEEEtran}

\usepackage{amsmath, amsfonts, amssymb}
\usepackage{algorithm, algorithmic}
\usepackage{graphicx, comment}
\usepackage{color,subfigure,cite}

\def\Rank{\mathop{\rm Rank}\limits}

\def\GF{\mathop{\sf GF}\limits}

\def\tran{\mathrm{T}}

\def\linsp{{\mathsf{span}}}

\def\X{{\mathbf{X}}}

\def\Z{{\mathbf{Z}}}

\def\vv{{\mathbf{v}}}

\def\vv{{\mathbf{v}}}
\def\w{{\mathbf{w}}}

\def\0{{\mathbf{0}}}
\def\1{{\mathbf{1}}}

\def\prop{{\mathsf{Prob}}}

\def\lmat{\left(\begin{matrix}}
\def\rmat{\end{matrix}\right)}
\def\eqref#1{(\ref{#1})}

\newenvironment{thmproof}[1]
{\noindent\hspace{2em}{\it #1 }}
{\hspace*{\fill}~\QED\par\endtrivlist\unskip}

\newtheorem{theorem}{Theorem}
\newtheorem{definition}{Definition}
\newtheorem{lemma}{Lemma}

\newtheorem{corollary}{Corollary}
\newtheorem{proposition}{Proposition}

\def\BibTeX{{\rm B\kern-.05em{\sc i\kern-.025em b}\kern-.08em
    T\kern-.1667em\lower.7ex\hbox{E}\kern-.125emX}}

\def\algtop#1{\vspace{.4cm}\hrule\vspace{.035cm}\hrule\vspace{.1cm}\noindent{\S~\sc #1}\par}
\def\algbot{\hrule\vspace{.035cm}\hrule\vspace{.4cm}}

\begin{document}

\title{Capacity of 1-to-$K$ Broadcast Packet Erasure Channels with Channel Output Feedback}

\author{\IEEEauthorblockN{Chih-Chun Wang}
\IEEEauthorblockA{
Center of Wireless Systems and Applications (CWSA)\\
School of Electrical and Computer Engineering, Purdue University, USA}
}

%

\maketitle

\begin{abstract}
This paper focuses on the 1-to-$K$ broadcast packet erasure channel (PEC), which is a generalization of the 
broadcast binary erasure channel from the binary symbol to that of arbitrary finite fields $\GF(q)$ with sufficiently large $q$. We consider the setting in which the source node has instant feedback of the channel outputs of the $K$ receivers after each transmission. Such a setting directly models {network coded packet transmission in the downlink direction} with integrated feedback mechanisms (such as Automatic Repeat reQuest (ARQ)).



 The main results of this paper are:  (i) The capacity region for general 1-to-3 broadcast PECs, and (ii) The capacity region for two classes of 1-to-$K$ broadcast PECs:
the {symmetric} PECs, and the {spatially independent} PECs with {one-sided fairness constraints}. This paper also develops
(iii) A pair of outer and inner bounds of the capacity region for arbitrary 1-to-$K$ broadcast PECs, which can be evaluated by any linear programming solver. For most practical scenarios, the outer and inner bounds meet and thus jointly characterize the capacity.


%
%
\end{abstract}


\section{Introduction}

In the last decade, the new network coding concept has emerged \cite{LiYeungCai03}, which focuses on achieving the capacity of a communication network. More explicitly, the network-coding-based approaches generally model each hop of a packet-based communication network by a {\em packet erasure channel} (PEC) instead of the classic Gaussian channel \cite{DanaGowaikarPalankiHassibiEffros06}. Such simple abstraction allows us to explore the information-theoretic capacity of a much larger network with mathematical rigor and also sheds new insights on the network effects of a communication system. One such example is the broadcast channel capacity with message side information. Unlike the existing Gaussian Broadcast Channel (GBC) results that are limited to the simplest 2-user scenario \cite{Wu07}, the capacity region for 1-to-$K$ broadcast PECs with message side information has been derived for $K=3$ and tightly bounded for general $K$ values \cite{Wang10a,WangKhreishahShroff09}.\footnote{The results of 1-to-$K$ broadcast PECs with message side information \cite{Wang10a,WangKhreishahShroff09} is related to the capacity of the wireless ``XOR-in-the-air" scheme \cite{KattiRahulHuKatabiMedardCrowcroft06}.} 
In addition to providing new insights on network communications, this simple PEC-based abstraction in network coding also accelerates the transition from theory to practice.
Many of the capacity-achieving {\em network codes} \cite{HoMedardKoetterKargerEffrosShiLeong06} have since been implemented for either the wireline \cite{ChouWuJain03} or the wireless multi-hop networks \cite{KattiRahulHuKatabiMedardCrowcroft06,KoutsonikolasWangHu10}.

Motivated by recent wireless network coding protocols, this paper studies the memoryless 1-to-$K$ broadcast PEC with Channel Output Feedback (COF). Namely, a single source node sends out a stream of packets wirelessly, which carries information of $K$ independent downlink data sessions, one for each receiver $d_k$, $k=1,\cdots, K$, respectively. After packet transmission through the broadcast PEC, each $d_k$ then informs the source its own channel output by sending back the ACKnowledgement (ACK) packets after each time slot.
\cite{GeorgiadisTassiulas09} derives the capacity region of the memoryless 1-to-2 broadcast PEC with COF. The results show that COF strictly improves the capacity of the memoryless 1-to-2 broadcast PEC, a mirroring result to the achievability results of GBCs with COF \cite{OzarowCheong84}.  Other than increasing the achievable throughput, COF can also be used for queue and delay management \cite{LiWangLin10,SundararajanShahMedard07} and for rate-control in a wireless network coded system \cite{KoutsonikolasWangHu10}.

The main contribution of this work includes: (i) The capacity region for general 1-to-3 broadcast PECs with COF; 
(ii) The capacity region for two classes of 1-to-$K$ broadcast PECs with COF:
the {\em symmetric} PECs, and the {\em spatially independent} PECs with {\em one-sided fairness constraints}; and
(iii) A pair of outer and inner bounds of the capacity region for general 1-to-$K$ broadcast PECs with COF, which can be evaluated by any linear programming solver. Extensive numerical experiments show that the outer and inner bounds meet for almost all practical scenarios and thus effectively bracket the capacity.

The capacity outer bound in this paper is derived by generalizing the degraded channel argument first proposed in \cite{OzarowCheong84}. For the achievability part of (i), (ii), and (iii), we devise a new class of inter-session network coded schemes, termed the {\em packet evolution method}. The packet evolution method is based on a novel concept of {\em network code alignment}, which is the PEC-counterpart of the interference alignment method originally proposed for Gaussian interference channels \cite{CadambeJafar08,DasVishwanathJafarMarkopoulou10}.


This paper is organized as follows. Section~\ref{sec:setting} contains the basic setting and the detailed comparison to the existing results. 
Section~\ref{sec:main} describes the main theorems of this paper. 
Section~\ref{sec:achievability} provides detailed description of the {\em packet evolution} scheme and the corresponding intuitions. Section~\ref{sec:achievability} also includes brief sketches on how to use the packet evolution method to prove the achivability results. (Most proofs of this paper are omitted due to the limit of space.)  Numerical evaluation is included in Section~\ref{subsec:sim}. Section~\ref{sec:conclusion} concludes this paper.


\section{Problem Setting \& Existing Results\label{sec:setting}}

\subsection{The Memoryless 1-to-$K$ Broadcast Packet Erasure Channel} 
For any positive integer $K$, we use $[K]\stackrel{\Delta}{=}\{1,2,\cdots,K\}$ to denote the set of integers from 1 to $K$, and use $2^{[K]}$ to denote the collection of all subsets of $[K]$.

Consider a 1-to-$K$ broadcast PEC from source $s$ to $K$ destinations $d_k$, $k\in [K]$. For each channel usage, the 1-to-$K$ broadcast PEC takes an input symbol $Y\in\GF(q)$ from $s$ and outputs a $K$-dimensional vector $\Z\stackrel{\Delta}{=}(Z_1,\cdots, Z_K)\in (\{Y\}\cup \{*\})^K$,  where the $k$-th coordinate $Z_k$ being ``$*$" denotes that the transmitted symbol $Y$ does not reach the $k$-th receiver $d_k$ (thus being erased). There is no other type of noise, 
i.e.,
the individual output is either equal to the input $Y$ or an erasure ``$*$." The {\em success probabilities} of a 1-to-$K$ PEC are described by $2^K$ non-negative parameters: $p_{S\overline{[K]\backslash S}}$ for all $S\in 2^{[K]}$ such that $\sum_{S\in 2^{[K]}}p_{S\overline{[K]\backslash S}}=1$ and for all $y\in \GF(q)$, 
\begin{align}
\prop\left(\left.\{k\in[K]:Z_k=y\}=S\right|Y=y\right)=p_{S\overline{[K]\backslash S}}.\nonumber
\end{align}
That is, $p_{S\overline{[K]\backslash S}}$ denotes the probability that the transmitted symbol $Y$ is received {\em by and only by} the receivers $\{d_k:k\in S\}$. 
For all $S\in 2^{[K]}$, we also define
\begin{align}
p_{\cup S}=\sum_{\forall S'\in 2^{[K]}: S'\cap S\neq \emptyset} p_{S'\overline{[K]\backslash S'}}.\nonumber
\end{align}
That is, $p_{\cup S}$ is the probability that {\em at least one of the receiver $d_k$ in $S$} successfully receives  $Y$.  We sometimes use $p_{k}$ as shorthand for $p_{\cup \{k\}}$, which is the marginal probability that the $k$-th receiver $d_k$ receives $Y$ successfully.

We assume that the broadcast PEC is {\em memoryless} and {\em time-invariant}, and use $Y(t)$ and $\Z(t)$ to denote the input and output for the $t$-th time slot.
Note that this setting allows the success events among different receivers to be dependent, also defined as {\em spatial dependence}. For example, when two logical receivers $d_{k_1}$ and $d_{k_2}$ are situated in the same physical node, we simply set the $p_{S\overline{[K]\backslash S}}$ parameters to allow perfect correlation between the success events of $d_{k_1}$ and $d_{k_2}$. Throughout this paper, we consider memoryless 1-to-$K$ broadcast PECs that may or may not be spatially dependent.


\subsection{Broadcast PEC Capacity with Channel Output Feedback}

We consider the following broadcast scenario from $s$ to $\{d_k:\forall k\in[K]\}$. Assume slotted transmission. Source $s$ is allowed to use the 1-to-$K$ PEC exactly $n$ times and would like to carry information for $K$ independent downlink data sessions, one for each $d_k$, respectively. For each $k\in[K]$, the $k$-th session (from $s$ to $d_k$) contains $nR_k$ information symbols $\X_k\stackrel{\Delta}{=}\{X_{k,j}\in\GF(q): \forall j\in [nR_k]\}$, where $R_k$ is the data rate for the $(s,d_k)$ session. All the information symbols $X_{k,j}$ for all $k\in[K]$ and $j\in[nR_k]$ are independently and uniformly distributed in $\GF(q)$.

We consider the setting with instant channel output feedback (COF). That is, for the $t$-th time slot, $s$ sends out a symbol
\begin{align}
Y(t)=
f_t\left(\{\X_k:\forall k\in [K]\}, \{\Z(\tau):\tau\in[t-1]\}\right),\nonumber
\end{align}
which is a function $f_t(\cdot)$ based on the information symbols $\{X_{k,j}\}$ and the COF $\{\Z(\tau):\tau\in[t-1]\}$ of the previous transmissions. In the end of the $n$-th time slot, each $d_k$ outputs the decoded symbols
\begin{align}
\hat{\X}_k\stackrel{\Delta}{=}\{\hat{X}_{k,j}:\forall j\in[nR_k]\}=g_k(\{Z_{k}(t):\forall t\in[n]\}),\nonumber
\end{align}
where $g_k(\cdot)$ is the decoding function of $d_k$ based on the corresponding observation $Z_k(t)$ for all $t\in[n]$. Note that we assume that the PEC channel parameters $\left\{p_{S\overline{[K]\backslash S}}:\forall S\in 2^{[K]}\right\}$ are available at $s$ before transmission. 

%

We now define the achievability of a 1-to-$K$  PEC with COF.

\begin{definition}A rate vector $(R_1,\cdots, R_K)$ is achievable if for any $\epsilon>0$, there exist sufficiently large $n$ and sufficiently large underlying finite field $\GF(q)$ such that
\begin{align}
\forall k\in[K],~\prop\left( \hat{\X}_k\neq \X_k\right)<\epsilon.\nonumber
\end{align}
\end{definition}
\begin{definition}The capacity region of a 1-to-$K$ PEC with COF is the closure of all achievable rate vectors.
\end{definition}

\subsection{Existing Results\label{subsec:existing}}
\begin{theorem}[Theorem~3 in \cite{GeorgiadisTassiulas09}] The capacity region $(R_1,R_2)$ of a 1-to-2 PEC with COF is described by
\begin{align}
\begin{cases}\frac{R_1}{p_{1}}+\frac{R_2}{p_{\cup\{1,2\}}}\leq 1\\
\frac{R_1}{p_{\cup\{1,2\}}}+\frac{R_2}{p_{2}}\leq 1
\end{cases}.\label{eq:2cap}\end{align}
\end{theorem}

One scheme that achieves the above capacity region in \eqref{eq:2cap} is the 2-phase approach in \cite{GeorgiadisTassiulas09}. That is, for any $(R_1,R_2)$ in the interior of \eqref{eq:2cap}, perform the following coding operations.

In Phase~1, $s$ sends out uncoded information packets $X_{1,j_1}$ and $X_{2,j_2}$ for all $j_1\in[nR_1]$ and $j_2\in[nR_2]$ until each packet is received by at least one receiver. Those $X_{1,j_1}$ packets that are received by $d_1$ have already reached their intended receiver and thus will not be retransmitted in the second phase. Those $X_{1,j_1}$ packets that are received by $d_2$ but not by $d_1$ need to be retransmitted in the second phase, and are thus stored in a separate queue $Q_{1;2\overline{1}}$. Symmetrically, the $X_{2,j_2}$ packets that are received by $d_1$ but not by $d_2$ need to be retransmitted, and are stored in another queue $Q_{2;1\overline{2}}$.  Since those ``overheard" packets in queues $Q_{1;2\overline{1}}$ and $Q_{2;1\overline{2}}$ are perfect candidates for intersession network coding \cite{KattiRahulHuKatabiMedardCrowcroft06}, they can be linearly mixed together in Phase~2. Each single coded packet in Phase~2 can now serve both $d_1$ and $d_2$ simultaneously. The intersession network coding gain in Phase~2 allows us to achieve the capacity region in \eqref{eq:2cap}.
Based on the same logic, \cite{LarssonJohansson06} derives an achievability region for 1-to-$K$ broadcast PECs with COF under a {\em perfectly symmetric setting}. 
\cite{RoznerIyerMehtaQiuJafry07} implements such 2-phase approach while taking into account of various practical considerations, such as time-out and network synchronization.

\subsection{The Suboptimality of The 2-Phase Approach\label{subsec:example}}

The above 2-phase approach does not achieve the capacity for the cases in which $K>2$. To illustrate this point, consider the example in Fig.~\ref{fig:example}.

\begin{figure}
\centering
\includegraphics[height=2.75cm]{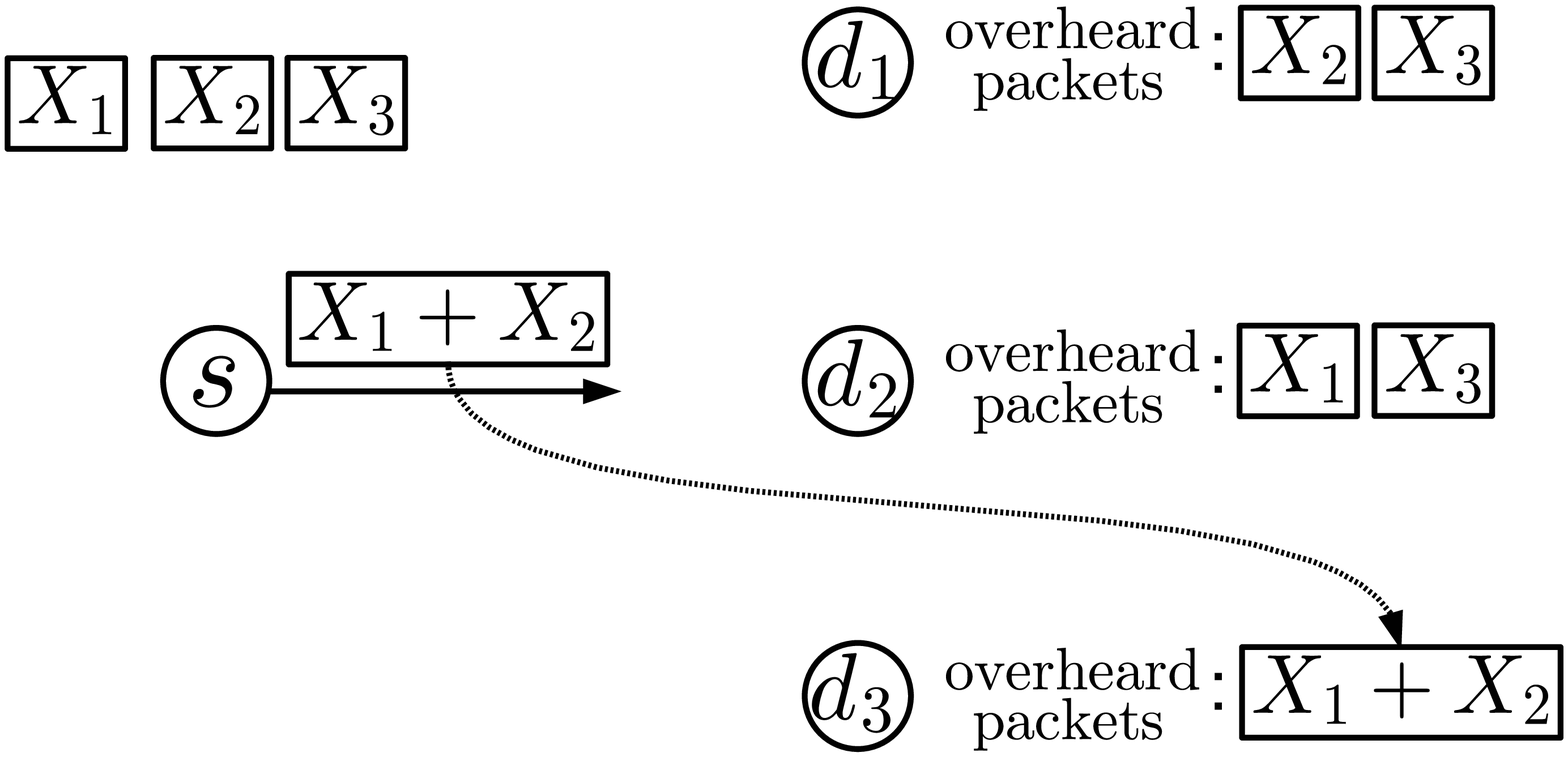}

%

\caption{Example of the suboptimality of the 2-phase approach. \label{subfig:example1}\label{fig:example}}
\end{figure}

In Fig.~\ref{subfig:example1}, source $s$ would like to serve three receivers $d_1$ to $d_3$. Each $(s,d_k)$ session contains a single information packet $X_k$, and the goal is to convey each $X_k$ to the intended $d_k$ for all $k=1,2,3$. Suppose the 2-phase approach in Section~\ref{subsec:existing} is used. During Phase~1, each packet is sent repeatedly until it is received by at least one receiver, which either conveys the packet to the intended receiver or creates an overheard packet that can be used in Phase~2. Suppose after Phase~1, $d_1$ has received $X_2$ and $X_3$, $d_2$ has received $X_1$ and $X_3$, and $d_3$ has not received any packet (Fig.~\ref{subfig:example1}). Since each packet has reached at least one receiver, source $s$ moves to Phase~2.

Suppose $s$ sends out a coded packet $[X_1+X_2]$ in Phase~2. Such coded packet can serve both $d_1$ and $d_2$. That is, $d_1$ (resp.\ $d_2$) can decode $X_1$ (resp.\ $X_2$)  by subtracting $X_2$ (resp.\ $X_1$) from $[X_1+X_2]$. Nonetheless, since the broadcast PEC is random, the packet $[X_1+X_2]$ may or may not reach $d_1$ or $d_2$. Suppose that due to random channel realization, $[X_1+X_2]$ reaches only $d_3$, see Fig.~\ref{subfig:example1}. The remaining question is what $s$ should send for the next time slot.

{\bf The existing 2-phase approach:} We first note that since $d_3$ received neither $X_1$ nor $X_2$ in the past, the newly received $[X_1+X_2]$ cannot be used by $d_3$ to decode any information packet. In the existing results \cite{LarssonJohansson06,GeorgiadisTassiulas09}, $d_3$ thus discards the overheard $[X_1+X_2]$, and $s$ would continue sending $[X_1+X_2]$ for the next time slot in order to capitalize this coding opportunity created in Phase~1.

{\bf The optimal decision:} It turns out that the broadcast system can actually benefit from the fact that $d_3$ overhears the coded packet $[X_1+X_2]$ even though neither $X_1$ nor $X_2$ can be decoded by $d_3$. More explicitly, instead of sending $[X_1+X_2]$, $s$ should send a new packet $[X_1+X_2+X_3]$ that mixes all three sessions together. With the new $[X_1+X_2+X_3]$ (plus the previous overhearing patterns in Fig.~\ref{subfig:example1}), $d_1$ can decode $X_1$ by subtracting both $X_2$ and $X_3$ from $[X_1+X_2+X_3]$. $d_2$ can decode $X_2$ by subtracting both $X_1$ and $X_3$ from $[X_1+X_2+X_3]$. For $d_3$, even though $d_3$ does not know the values of  $X_1$ and $X_2$, $d_3$ can still use the previously overheard $[X_1+X_2]$ packet to subtract the interference $(X_1+X_2)$ from $[X_1+X_2+X_3]$ and decode its desired packet $X_3$. As a result, the new coded packet $[X_1+X_2+X_3]$ serves  $d_1$, $d_2$, and $d_3$, simultaneously. This new coding decision thus strictly outperforms the existing 2-phase approach.

Two critical observations can be made for this example. First of all, when $d_3$ overhears a coded $[X_1+X_2]$ packet, even though $d_3$ can decode neither $X_1$ nor $X_2$, such new side information can still be used for future decoding. More explicitly, as long as $s$ sends packets that are of the form $\alpha(X_1+X_2)+\beta X_3$, the ``aligned interference" $\alpha(X_1+X_2)$ can be completely removed by $d_3$ without decoding individual $X_1$ and $X_2$. This technique is thus termed  ``{\em code alignment}," which is in parallel with the original interference alignment method \cite{CadambeJafar08}. Second of all, in the existing 2-phase approach, Phase~1 has the dual roles of sending uncoded packets to their intended receivers, and, at the same time,  creating new coding opportunities (the overheard packets) for Phase~2. It turns out that this dual-purpose Phase-1 operation is indeed optimal.  The suboptimality of the 2-phase approach for $K>2$ is actually caused by the Phase-2 operation, in which source $s$ only capitalizes the coding opportunities created in Phase~1 but does not create any new coding opportunities for subsequent packet mixing. One can thus envision that for the cases $K>2$, an optimal policy should be a multi-phase policy, say an $M$-phase policy, such that for all $i\in[M-1]$ (not only for the first phase) the coded packets sent in the $i$-th phase have dual roles of carrying information to their intended receivers and simultaneously creating new coding opportunities for the subsequent Phases $(i+1)$ to $M$. These two observations will be the building blocks of our achievability results. 


\section{The Main Results\label{sec:main}}


Section~\ref{subsec:general} focuses on the capacity results for arbitrary broadcast PEC parameters while Section~\ref{subsec:special} considers two special classes of broadcast PECs: the symmetric and the spatially independent PECs, respectively.

\subsection{Capacity Results For General 1-to-$K$ Broadcast PECs\label{subsec:general}}

We define any bijective function $\pi:[K]\mapsto[K]$ as a permutation. 
There are totally $K!$ distinct permutations.
Given any permutation $\pi$, for all $j\in [K]$ we define $S^\pi_j\stackrel{\Delta}{=}\{\pi(l):\forall l\in[j]\}$ as the set of the first $j$ elements according to the permutation $\pi$. We then have the following capacity outer bound for any 1-to-$K$ broadcast PEC with COF.

\begin{proposition}\label{prop:outer}
Any achievable rates  $(R_1,\cdots, R_K)$ must satisfy the following $K!$ inequalities:
\begin{align}
\forall \pi,~\sum_{j=1}^K\frac{R_{\pi(j)}}{p_{\cup S^\pi_j}}\leq 1.\label{eq:pi-outer}
\end{align}
\end{proposition}


\begin{thmproof}{Sketch of the proof:}
For any given $\pi$, construct a new broadcast channel from the original one by adding $(K-1)$ information pipes connecting all the receivers $d_1$ to $d_K$. More explicitly, for all $j\in[K-1]$, create an auxiliary pipe from $d_{\pi(j)}$ to $d_{\pi(j+1)}$. With the new auxiliary pipes, the success probability of $d_{\pi(j)}$ increases from $p_{\pi(j)}$ to $p_{\cup S_j^\pi}$ for all $j\in[K]$ since $d_{\pi(j)}$ now knows the transmitted symbol $Y$ as long as at least one of $d_{\pi(l)}$, $\forall l\in[j]$, receives $Y$ successfully. Note that the new broadcast PEC is physically degraded. By the same arguments as in \cite{OzarowCheong84,GeorgiadisTassiulas09}, \eqref{eq:pi-outer} describes the capacity of a physically degraded PEC with COF, which thus outer bounds the capacity of the original PEC with COF.


\end{thmproof}

For the following, we provide the capacity results for general 1-to-3 broadcast PECs.

\begin{proposition}\label{prop:cap3}
For any parameter values $\left\{p_{S\overline{\{1,2,3\}\backslash S}}:\forall S\in 2^{\{1,2,3\}}\right\}$ of a 1-to-3 PEC with COF, the capacity outer bound in Proposition~\ref{prop:outer} is the capacity region. 
\end{proposition}

To state the capacity inner bound for 1-to-$K$ PECs with $K\geq 4$, we need to define an additional function: $f_p(S\overline{T})$, which takes an input $S\overline{T}$ of two disjoint sets $S,T\in 2^{[K]}$. More explicitly, $f_p(S\overline{T})$ is the probability that the transmitted packet $Y$ is received by all those $d_i$ with $i\in S$ but not received by any $d_j$ with $j\in T$. That is,
\begin{align}
&f_p(S\overline{T})\stackrel{\Delta}{=}\sum_{\forall S_1: S\subseteq S_1, T\subseteq ([K]\backslash S_1)}p_{S_1\overline{[K]\backslash S_1}}.\nonumber
\end{align}
We also say that a {\em strict total ordering} ``$\prec$" on $2^{[K]}$ is {\em cardinality-compatible} if
\begin{align}
\forall S_1, S_2\in 2^{[K]},\quad |S_1|<|S_2|\Rightarrow S_1\prec S_2.\nonumber
\end{align}
For example, for $K=3$, the following strict total ordering
\begin{align}
\emptyset\prec \{2\}\prec\{1\}\prec\{3\}\prec \{1,2\}\prec\{1,3\}\prec\{2,3\}\prec\{1,2,3\}\nonumber
\end{align}
is cardinality-compatible. 

\begin{proposition}\label{prop:ach2} Fix any cardinality-compatible, strict total ordering $\prec$.
For any 1-to-$K$ PEC with COF, a rate vector $(R_1,\cdots, R_K)$ can be achieved by a {\em linear network code}
if
there exist $2^K$ non-negative $x$ variables, indexed by $S\in 2^{[K]}$:
\begin{align}
\left\{x_S\geq 0:\forall S\in 2^{[K]}\right\},\label{eq:xs}
\end{align}
and $K3^{K-1}$ non-negative $w$ variables, indexed by $(k;S\rightarrow T)$ satisfying $T\subseteq S\subseteq ([K]\backslash k)$:
\begin{align}&\left\{w_{k;S\rightarrow T}\geq 0: \forall k\in[K],\forall S,T\in 2^{[K]},\right.\nonumber\\
&\hspace{3cm}\left.\text{satisfying }T\subseteq S\subseteq ([K]\backslash k)\right\},\label{eq:ws}
\end{align}
such that jointly the following linear inequalities\footnote{
There are totally $(1+K2^{K-1}+K3^{K-1})$ inequalities. More explicitly, \eqref{eq:total-x} describes one inequality. There are $K2^{K-1}$ inequalities having the form of \eqref{eq:coding-len}. There are totally $K3^{K-1}$ inequalities having the form of one of \eqref{eq:ind-length-0}, \eqref{eq:ind-length-1}, and \eqref{eq:ind-length-2}. For comparison, the outer bound in Proposition~\ref{prop:outer} actually has more inequalities asymptotically ($K!$ of them) than those in Proposition~\ref{prop:ach2}.} are satisfied:
\begin{align}
&\sum_{\forall S:S\in 2^{[K]}}x_S< 1\label{eq:total-x}\\
&\forall T\in 2^{[K]},\forall k\in T,\nonumber\\
&\hspace{2.5cm} x_T\geq \sum_{\forall S:(T\backslash k)\subseteq S\subseteq([K]\backslash k)}w_{k;S\rightarrow (T\backslash k)}\label{eq:coding-len}\\
&\forall k\in[K],\quad w_{k;\emptyset\rightarrow \emptyset}\cdot p_{\cup [K]}\geq R_k\label{eq:ind-length-0}
\end{align}
\begin{align}
&\forall k\in[K], \forall S\subseteq ([K]\backslash k), S\neq \emptyset, \nonumber\\
&\hspace{0cm}
\left(\sum_{\forall T_1: T_1\subseteq S} w_{k;S\rightarrow T_1}\right) p_{\cup ([K]\backslash S)}\geq \nonumber\\
&\hspace{.7cm}\sum_{\scriptsize \begin{array}{c}\forall S_1,T_1:\text{such that}\\
T_1\subseteq S_1\subseteq ([K]\backslash k),\\
T_1\subseteq S,S\nsubseteq S_1 \end{array}}w_{k;S_1\rightarrow T_1}\cdot f_p\left((S\backslash T_1)\overline{([K]\backslash S)}\right)\label{eq:ind-length-1}
\end{align}
\begin{align}
&\forall k\in[K], S,T\in 2^{[K]} \text{ satisfying } T\subseteq S\subseteq ([K]\backslash k), T\neq S,\nonumber\\
&\hspace{0cm}\left(w_{k;S\rightarrow T}+\sum_{\scriptsize\begin{array}{c}\forall T_1\subseteq S:\\
(T_1\cup\{k\})\prec (T\cup\{k\})\end{array}}w_{k; S\rightarrow T_1}\right)p_{\cup ([K]\backslash S)} \leq\nonumber\\
&\hspace{.5cm}
\sum_{\scriptsize\begin{array}{c}\forall S_1: S_1\prec S,\\
T\subseteq S_1\subseteq ([K]\backslash k)\end{array}}w_{k;S_1\rightarrow T} \cdot f_p\left((S\backslash T)\overline{([K]\backslash S)}\right)+\nonumber\\
&\hspace{.7cm}\sum_{\scriptsize \begin{array}{c}\forall S_1,T_1:\text{such that}\\
T_1\subseteq S_1\subseteq ([K]\backslash k),\\
(T_1\cup\{k\})\prec (T\cup\{k\}),\\
T_1\subseteq S,S\nsubseteq S_1 \end{array}}w_{k;S_1\rightarrow T_1}\cdot f_p\left((S\backslash T_1)\overline{([K]\backslash S)}\right).\label{eq:ind-length-2}
\end{align}
\end{proposition}

%
%
%
{\em Remark:} For some general classes of PEC parameters, one can prove that the inner bound of Proposition~\ref{prop:ach2} is indeed the capacity region for arbitrary $K\geq 4$ values. Two such classes are discussed in the next subsection.

\subsection{Capacity Results For Two Classes of 1-to-$K$ PECs\label{subsec:special}}

We first focus on {\em symmetric} broadcast PECs.
\begin{definition}
A 1-to-$K$ broadcast PEC is {\em symmetric} if the channel parameters $\left\{p_{S\overline{[K]\backslash S}}:\forall S\in 2^{[K]}\right\}$ satisfy
\begin{align}
\forall S_1,S_2\in 2^{[K]}\text{ with }|S_1|=|S_2|,~p_{S_1\overline{[K]\backslash S_1}}=p_{S_2\overline{[K]\backslash S_2}}.\nonumber
\end{align}
\end{definition}

\begin{proposition}\label{prop:cap-sym}
For any symmetric 1-to-$K$ broadcast PEC with COF, the  capacity outer bound in Proposition~\ref{prop:outer} is indeed the corresponding capacity region.
\end{proposition}

In addition to perfect channel symmetry, another practical setting is to allow channel asymmetry while assuming {\em spatial independence} between different destinations $d_i$. 

\begin{definition}
A 1-to-$K$ broadcast PEC is {\em spatially independent} if the channel parameters $\left\{p_{S\overline{[K]\backslash S}}:\forall S\in 2^{[K]}\right\}$ satisfy
\begin{align}
\forall S\in2^{[K]},~ p_{S\overline{[K]\backslash S}}=\left(\prod_{k\in S}p_k\right)\left(\prod_{k\in [K]\backslash S} (1-p_k)\right),\nonumber
\end{align}
where $p_k$ is the marginal success probability of destination $d_k$.
\end{definition}


To describe the capacity results for spatially independent 1-to-$K$ PECs, we need the following additional definition.

\begin{definition}
Consider a 1-to-$K$ broadcast PEC with marginal success probabilities $p_1$ to $p_K$. 
We say a rate vector $(R_1,\cdots, R_K)$ is {\em one-sidedly fair} if
$\forall i\neq j$ satisfying $p_i\leq p_j$, we have $R_i(1-p_i)\geq R_j(1-p_j)$. We use $\Lambda_{\text{osf}}$ to denote the collection of all one-sidedly fair rate vectors.
\end{definition}

The one-sided fairness contains many practical scenarios of interest. For example, the perfectly fair rate vector $(R,R,\cdots, R)$ by definition is also one-sidedly fair. Another example is when $\min(p_1,\cdots, p_K)\geq \frac{1}{2}$, a proportionally fair rate vector 
$(p_1R, p_2R,\cdots, p_KR)$ is also one-sidedly fair.

For the following, we provide the capacity of spatially independent 1-to-$K$ PECs with COF under the condition of one-sided fairness.

\begin{proposition}\label{prop:cap-osf}
Suppose the 1-to-$K$ PEC of interest is spatially independent and the marginal success probabilities satisfy $0<p_1\leq p_2\leq \cdots\leq p_K$, which can be achieved by relabeling. Any $(R_1,\cdots, R_K)\in \Lambda_\text{osf}$ is in the capacity region if and only if $(R_1,\cdots, R_K)\in \Lambda_\text{osf}$ satisfies
\begin{align}
\sum_{k=1}^K\frac{R_k}{1-\prod_{l=1}^k(1-p_l)}\leq  1.\label{eq:osf-cap}
\end{align}
\end{proposition}

Namely, Proposition~\ref{prop:outer} is indeed the capacity region when focusing on the one-sidedly fair rate region $\Lambda_\text{osf}$.


\section{The Packet Evolution Schemes\label{sec:achievability}}

We now describe a new class of coding schemes, termed the {\em packet evolution} (PE) scheme, which 
is the building block of the capacity /  achievability results in Section~\ref{sec:main}.


\subsection{Description Of The Packet Evolution Scheme\label{subsec:PE}}
Recall that each $(s,d_k)$ session has $nR_k$ information packets $X_{k,1}$ to $X_{k,nR_k}$. We associate each of the $\sum_{k=1}^KnR_k$ information packets  with {\em an intersession coding vector} $\vv$ and a set $S\subseteq [K]$. An intersession coding vector is a $\left(\sum_{k=1}^KnR_k\right)$-dimensional row vector with each coordinate being a scalar in $\GF(q)$. Before the start of the broadcast, for any $k\in[K]$ and $j\in[nR_k]$ we initialize the corresponding vector $\vv$ of $X_{k,j}$ in a way that the only nonzero coordinate of $\vv$ is the coordinate corresponding to $X_{k,j}$ and all other coordinates are zero. Without loss of generality, we set the value of the only non-zero coordinate to one. That is, initially the coding vectors $\vv$ are set to the elementary basis vectors. 

 For any $k\in[K]$ and $j\in[nR_k]$ the set $S$ of $X_{k,j}$ is initialized to $\emptyset$. We call $S$ the {\em overhearing set} 
 of the packet $X_{k,j}$. We use $\vv(X_{k,j})$ and $S(X_{k,j})$ to denote the intersession coding vector and the overhearing set of a given $X_{k,j}$.

Throughout the $n$ broadcast time slots, $s$ constantly updates $S(X_{k,j})$ and $\vv(X_{k,j})$ according to the COF. The main structure of a packet evolution scheme can now be described as follows.

\algtop{The Packet Evolution Scheme}
\begin{algorithmic}[1]

\STATE Source $s$ maintains a flag ${\mathsf{f}}_{\text{change}}$. Initially, set ${\mathsf{f}}_{\text{change}}\leftarrow 1$.
\FOR{$t=1,\cdots, n$, }
\STATE In the beginning of the $t$-th time slot, do Lines~\ref{line:begin1} to~\ref{line:transmit-v}.
\IF{${\mathsf{f}}_{\text{change}}= 1$\label{line:begin1}}
\STATE Choose a non-empty subset $T\subseteq [K]$.
\STATE Run a subroutine {\sc Packet Selection}, which takes $T$ as input and outputs a collection of $|T|$ packets $\{X_{k,j_k}: \forall k\in T\}$, termed the {\em target packets}. The output $\{X_{k,j_k}\}$ must satisfy $(S(X_{k,j_k})\cup \{k\})\supseteq T$ for all $k\in T$.\label{line:targeting-set}

\STATE Generate $k$ uniformly random coefficients $c_k\in \GF(q)$ and construct an intersession coding vector $\vv_\text{tx}\leftarrow \sum_{k\in T}c_k\cdot \vv(X_{k,j_k})$.\label{line:vv-construct}

\STATE Set ${\mathsf{f}}_{\text{change}}\leftarrow 0$.
\ENDIF

\STATE Sends out a linearly intersession coded packet according to the coding vector $\vv_\text{tx}$. That is, we send
\begin{align}Y_\text{tx}=\vv_\text{tx}\cdot (X_{1,1},\cdots, X_{K,nR_K})^\tran\nonumber
\end{align}  where $(X_{1,1},\cdots, X_{K,nR_K})^\tran$ is a column vector consisting of all information symbols.
\label{line:transmit-v}



\STATE In the end of the $t$-th time slot, use a subroutine {\sc Update} to revise the $\vv(X_{k,j_k})$ and $S(X_{k,j_k})$ values of all target packets $X_{k,j_k}$ based on the COF.

\IF{the $S(X_{k,j_k})$ value changes for at least one target packet $X_{k,j_k}$ after the {\sc Update}}
\STATE Set ${\mathsf{f}}_{\text{change}}\leftarrow 1$.
\ENDIF

\ENDFOR

\end{algorithmic} \algbot

In summary, a group of target packets $\{X_{k,j_k}\}$ are selected according to the choice of the subset $T$. The corresponding vectors $\{\vv(X_{k,j_k})\}$ are used to construct a coding vector $\vv_\text{tx}$. The same coded packet $Y_\text{tx}$, corresponding to $\vv_\text{tx}$, is then sent repeatedly for many time slots until one of the target packets $X_{k,j_k}$ {\em evolves} (when the corresponding $S(X_{k,j_k})$ changes). Then a new subset $T$ is chosen and the process is repeated until we use up all $n$ time slots. Three subroutines are used as the building blocks of a packet evolution method: (i) How to choose the non-empty $T\subseteq [K]$; (ii) For each $k\in[K]$, how to select a single target packets $X_{k,j_k}$ among all $X_{k,j}$ satisfying $(S(X_{k,j})\cup \{k\})\supseteq T$; and (iii) How to update the coding vectors $\vv(X_{k,j_k})$ and the overhearing sets $S(X_{k,j_k})$. We first describe the detailed update rule of (iii).

\algtop{Update of $S(X_{k,j_k})$ and $\vv(X_{k,j_k})$}
\begin{algorithmic}[1]
\STATE {\bf Input:} The $T$ and $\vv_\text{tx}$ used for transmission in the current time slot; And $S_\text{rx}$, the set of destinations $d_i$ that receive the transmitted coded packet in the current time slot. 
\FOR{all $k\in T$}\label{line:update-4-K}
\IF{$S_\text{rx}\nsubseteq S(X_{k,j_k})$}
\STATE Set $S(X_{k,j_k})\leftarrow (T\cap S(X_{k,j_k}))\cup S_\text{rx}$.\label{line:S-update}
 \STATE Set $\vv(X_{k,j_k})\leftarrow \vv_\text{tx}$.\label{line:v-update}
\ENDIF

\ENDFOR

\end{algorithmic} \algbot

\vspace{.2cm}
\noindent {\em An Illustrative Example Of The PE Scheme:}

\vspace{.2cm}
Let us revisit the optimal coding scheme of the example in Fig.~\ref{fig:example} of Section~\ref{subsec:example}. After initialization, the three information packets $X_{1}$ to $X_3$ have the corresponding $\vv$ and $S$: $\vv(X_1)=(1,0,0)$, $\vv(X_2)=(0,1,0)$, and $\vv(X_3)=(0,0,1)$, and $S(X_1)=S(X_2)=S(X_3)=\emptyset$. We use the following table
for summary.

\begin{center}\begin{tabular}{|c|c|c|}
\hline
$X_1$: (1,0,0),$\emptyset$ &$X_2$: (0,1,0),$\emptyset$ & $X_3$: (0,0,1),$\emptyset$ \\
\hline
\end{tabular}
\end{center}

Consider a duration of 5 time slots.

Slot 1: Suppose that $s$ chooses $T=\{1\}$. Since $(\emptyset\cup \{1\})\supseteq T$, {\sc Packet Selection} outputs $X_1$. The coding vector $\vv_\text{tx}$ is thus a scaled version of $\vv(X_1)=(1,0,0)$. Without loss of generality, we choose $\vv_\text{tx}=(1,0,0)$. Based on $\vv_\text{tx}$, $s$ transmits a packet $1 X_1+0 X_2+0 X_3=X_1$. Suppose $[X_1]$ is received by $d_2$, i.e., $S_\text{rx}=\{2\}$. Then during {\sc Update}, $S_\text{rx}=\{2\}\nsubseteq S(X_1)=\emptyset$. {\sc Update} thus sets $S(X_1)=\{2\}$ and $\vv(X_1)=\vv_\text{tx}=(1,0,0)$. The packet summary becomes

\begin{center}\begin{tabular}{|c|c|c|}
\hline
$X_1$: (1,0,0),$\{2\}$ &$X_2$: (0,1,0),$\emptyset$ & $X_3$: (0,0,1),$\emptyset$ \\
\hline
\end{tabular}.
\end{center}

Slot 2: Suppose that $s$ chooses $T=\{2\}$. Since $(\emptyset\cup \{2\})\supseteq T$,  {\sc Packet Selection} outputs $X_2$. The coding vector $\vv_\text{tx}$ is thus a scaled version of $\vv(X_2)=(0,1,0)$. Without loss of generality, we choose $\vv_\text{tx}=(0,1,0)$ and accordingly $[X_2]$ is sent. Suppose $[X_2]$ is received by $d_1$, i.e., $S_\text{rx}=\{1\}$. Since $S_\text{rx}\nsubseteq S(X_2)$, after {\sc Update} the packet summary becomes

\begin{center}\begin{tabular}{|c|c|c|}
\hline
$X_1$: (1,0,0),$\{2\}$ &$X_2$: (0,1,0),$\{1\}$ & $X_3$: (0,0,1),$\emptyset$ \\
\hline
\end{tabular}.
\end{center}

Slot 3: Suppose that $s$ chooses $T=\{3\}$ and {\sc Packet Selection} outputs $X_3$. $\vv_\text{tx}$ is thus a scaled version of $\vv(X_3)=(0,0,1)$, and we choose $\vv_\text{tx}=(0,0,1)$. Accordingly $[X_3]$ is sent. Suppose $[X_3]$ is received by $\{d_1,d_2\}$, i.e., $S_\text{rx}=\{1,2\}$. Then after {\sc Update}, the summary becomes

\begin{center}\begin{tabular}{|c|c|c|}
\hline
$X_1$: (1,0,0),$\{2\}$ &$X_2$: (0,1,0),$\{1\}$ & $X_3$: (0,0,1),$\{1,2\}$ \\
\hline
\end{tabular}.
\end{center}

Slot 4: Suppose that $s$ chooses $T=\{1,2\}$. Since $(S(X_1)\cup \{1\})\supseteq T$ and $(S(X_2)\cup \{2\})\supseteq T$, {\sc Packet Selection} outputs $\{X_1,X_2\}$. $\vv_\text{tx}$ is thus a linear combination of $\vv(X_1)=(1,0,0)$ and $\vv(X_2)=(0,1,0)$. Without loss of generality, we choose $\vv_\text{tx}=(1,1,0)$ and accordingly $[X_1+X_2]$ is sent. Suppose $[X_1+X_2]$ is received by $d_3$, i.e., $S_\text{rx}=\{3\}$. Then during {\sc Update}, for $X_1$,  $S_\text{rx}=\{3\}\nsubseteq S(X_1)=\{2\}$. {\sc Update} thus sets $S(X_1)=\{2,3\}$ and $\vv(X_1)=\vv_\text{tx}=(1,1,0)$. For $X_2$, $S_\text{rx}=\{3\} \nsubseteq S(X_2)=\{1\}$. {\sc Update} thus sets $S(X_2)=\{1,3\}$ and $\vv(X_2)=\vv_\text{tx}=(1,1,0)$.  The summary becomes


\begin{center}\begin{tabular}{|c|c|}
\hline
$X_1$: (1,1,0),$\{2,3\}$ &$X_2$: (1,1,0),$\{1,3\}$\\
\hline
 $X_3$: (0,0,1),$\{1,2\}$  &~ \\
\hline
\end{tabular}.
\end{center}


Slot 5: Suppose that $s$ chooses $T=\{1,2,3\}$. By Line~\ref{line:targeting-set} of {\sc The Packet Evolution Scheme}, the subroutine {\sc Packet Selection} outputs $\{X_1,X_2,X_3\}$. $\vv_\text{tx}$ is thus a linear combination of $\vv(X_1)=(1,1,0)$, $\vv(X_2)=(1,1,0)$, and $\vv(X_3)=(0,0,1)$, which is of the form $\alpha(X_1+X_2)+\beta X_3$. {\em Note that the packet evolution scheme automatically achieves code alignment}, which is the key component of the optimal coding policy in Section~\ref{subsec:example}. Without loss of generality, we choose $\alpha=\beta=1$ and $\vv_\text{tx}=(1,1,1)$. $Y_\text{tx}=[X_1+X_2+X_3]$ is sent accordingly. Suppose $[X_1+X_2+X_3]$ is received by $\{d_1,d_2,d_3\}$, i.e., $S_\text{rx}=\{1,2,3\}$. Then after {\sc Update}, the summary of the packets becomes
\begin{center}\begin{tabular}{|c|c|}
\hline
$X_1$: (1,1,1),$\{1,2,3\}$ &$X_2$: (1,1,1),$\{1,2,3\}$\\
\hline
  $X_3$: (1,1,1),$\{1,2,3\}$ &~ \\
\hline
\end{tabular}.
\end{center}

From the above step-by-step illustration, we see that the optimal coding policy in Section~\ref{subsec:example} is a special case of a packet evolution scheme.

\subsection{Properties of A Packet Evolution Scheme}

We term the packet evolution (PE) scheme in Section~\ref{subsec:PE} a {\em generic} PE method since it does not depend on how to choose $T$ and the target packets $X_{k,j_k}$ and only requires the output of {\sc Packet Selection} satisfying
$(S(X_{k,j_k})\cup \{k\})\supseteq T, \forall k\in T$. In this subsection, we state some key properties of any generic PE scheme. The intuition of the PE scheme is based on these key properties and will be discussed in Section~\ref{subsec:intuition-PE}.

We first define the following notation for any linear network codes. (Note that the PE scheme is a linear network code.)


\begin{definition}
Consider any linear network code. For any destination $d_k$, each of the received packet $Z_k(t)$ can be represented by a vector $\w_k(t)$, which is a $\left(\sum_{k=1}^KnR_k\right)$-dimensional vector containing the coefficients used to generate $Z_k(t)$. That is, $Z_k(t)=\w_k(t)\cdot (X_{1,1},\cdots, X_{K,nR_K})^\tran$. If $Z_k(t)$ is an erasure, we simply set $\w_k(t)$ to be an all-zero vector. The {\em knowledge space} of destination $d_k$ in the end of time $t$ is denoted by $\Omega_{\text{Z},k}(t)$, which is the linear span of $\w_k(\tau)$, $\tau\leq t$. That is, $\Omega_{\text{Z},k}(t)\stackrel{\Delta}{=}\linsp(\w_k(\tau):\forall \tau\in[t])$.
\end{definition}

\begin{definition} For any non-coded information packet $X_{k,j}$, the corresponding intersession coding vector is a $\left(\sum_{k=1}^KnR_k\right)$-dimensional vector with a single one in the corresponding coordinate and all other coordinates being zero. We use $\delta_{k,j}$ to denote such a delta vector. The message space of $d_k$ is then defined as $\Omega_{M,k}=\linsp(\delta_{k,j}:\forall j\in[nR_k])$.
\end{definition}

The above definitions imply the following straightforward lemma:
\begin{lemma}\label{lem:simple-dec}
In the end of time $t$, destination $d_k$ is able to decode all the desired information packets $X_{k,j}$, $\forall j\in [nR_k]$, if and only if $\Omega_{M,k}\subseteq \Omega_{Z,k}(t)$.
\end{lemma}

We now define ``non-interfering vectors" from the perspective of a destination $d_k$.
\begin{definition}
In the end of time $t$ (or in the beginning of time $(t+1)$), a vector $\vv$ (and thus the corresponding coded packet) is ``non-interfering" from the perspective of $d_k$ if
\begin{align}
\vv\in \linsp(\Omega_{Z,k}(t),\Omega_{M,k}).\nonumber
\end{align}
\end{definition}

By definition, any non-interfering vector $\vv$ can always be expressed as the sum of two vectors $\vv'$ and $\w$, where $\vv'\in \Omega_{M,k}$ is a linear combination of all information vectors for $d_k$ and $\w\in \Omega_{Z,k}(t)$ is a linear combination of all the packets received by $d_k$. If $\vv'=0$, then $\vv=\w$ is a {\em transparent} packet from $d_k$'s perspective since $d_k$ can compute the value of $\w\cdot (X_{1,1},\cdots, X_{K,nR_K})^\tran$ from its current knowledge space $\Omega_{Z,k}(t)$.
If $\vv'\neq 0$, then $\vv=\vv'+\w$ can be viewed as a pure information packet $\vv'\in \Omega_{M,k}$ after subtracting the unwanted $\w$ vector. In either case, $\vv$ is {\em not interfering} with the transmission of the $(s,d_k)$ session, which gives the name of ``non-interfering vectors."

The following Lemmas~\ref{lem:non-interfering} and \ref{lem:decodability} discuss the time dynamics of the PE scheme. To distinguish different time instants, we add a time subscript and use $S_{t-1}(X_{k,j_k})$ and $S_{t}(X_{k,j_k})$ to denote the overhearing set of $X_{k,j_k}$ in the end of time $(t-1)$ and $t$, respectively. Similarly, $\vv_{t-1}(X_{k,j_k})$ and $\vv_{t}(X_{k,j_k})$ denote the coding vectors in the end of time $(t-1)$ and $t$, respectively.

\begin{lemma}\label{lem:non-interfering}In the end of the $t$-th time slot, consider any $X_{k,j}$ out of all the information packets $X_{1,1}$ to $X_{K,nR_K}$. Its assigned vector $\vv_t(X_{k,j})$ is non-interfering from the perspective of $d_i$ for all $i\in (S_t(X_{k,j})\cup\{k\})$.
\end{lemma}


To illustrate Lemma~\ref{lem:non-interfering}, consider our 5-time-slot example. In the end of Slot~4, we have $\vv(X_1)=(1,1,0)$ and $S(X_1)\cup \{1\}=\{1,2,3\}$. It can be easily verified by definition that $\vv(X_1)=(1,1,0)$ is non-interfering from the perspectives of $d_1$, $d_2$, and $d_3$, respectively.


\begin{lemma}\label{lem:decodability} In the end of the $t$-th time slot, we use $\Omega_{R,k}(t)$ to denote the {\em remaining space} of the PE scheme:
\begin{align}
&\Omega_{R,k}(t)\stackrel{\Delta}{=}\nonumber\\
&\linsp(\vv_t(X_{k,j}):\forall j\in [nR_k] \text{ satisfying } k\notin S_t(X_{k,j})).\nonumber
\end{align}

For any $n$ and any $\epsilon>0$, there exists a sufficiently large finite field $\GF(q)$ such that for all $k\in[K]$ and $t\in[n]$,
\begin{align}
&\prop\left(\linsp(\Omega_{Z,k}(t), \Omega_{R,k}(t))=\linsp(\Omega_{Z,k}(t), \Omega_{M,k})\right)\nonumber\\
&>1-\epsilon.\nonumber
\end{align}
\end{lemma}

Intuitively, Lemma~\ref{lem:decodability} says that if in the end of time $t$ we directly transmit all the {\em remaining} coded packets $\left\{\vv_t(X_{k,j}):\forall j\in [nR_k], k\notin S_t(X_{k,j})\right\}$ from $s$ to $d_k$ through a noise-free information pipe, then with high probability, $d_k$ can successfully decode all the desired information packets $X_{k,1}$ to $X_{k,nR_k}$ (see Lemma~\ref{lem:simple-dec}) by the knowledge space $\Omega_{Z,k}(t)$ and the new information of the remaining space $\Omega_{R,k}(t)$.

\subsection{The Intuitions Of The Packet Evolution Scheme\label{subsec:intuition-PE}}
Lemmas~\ref{lem:non-interfering} and~\ref{lem:decodability} are the key properties of a PE scheme. In this subsection, we discuss the corresponding intuitions.

{\sf Receiving} {\bf the information packet $X_{k,j}$:}\quad
Each information packet keeps a coding vector $\vv(X_{k,j})$. Whenever we would like to communicate $X_{k,j}$ to destination $d_k$, instead of sending a non-coded packet $X_{k,j}$ directly, the PE scheme sends an intersession coded packet according to the coding vector $\vv(X_{k,j})$. Lemma~\ref{lem:decodability} shows that if we send all the coded vectors $\vv(X_{k,j})$ that have not been heard by $d_k$ (with $k\notin S(X_{k,j})$) through a noise-free information pipe, then $d_k$ can indeed decode all the desired packets $X_{k,j}$ with close-to-one probability. It also implies, although in an implicit way, that once a $\vv(X_{k,j_0})$ is heard by $d_k$ for some $j_0$ (therefore $k\in S(X_{k,j_0})$), there is no need to transmit this particular $\vv(X_{k,j_0})$ in the later time slots. Jointly, these two implications show that we can indeed use the coded packet $\vv(X_{k,j})$ as a substitute for $X_{k,j}$ without losing any information. In the broadest sense, we can say that $d_k$ {\sf receives} a packet $X_{k,j}$ if the corresponding $\vv(X_{k,j})$ successfully arrives $d_k$ in some time slot $t$.


{\bf Serving multiple destinations simultaneously by mixing non-interfering packets:}\quad The above discussion ensures that when we would like to send 
$X_{k,j_k}$ to $d_k$, we can send a coded packet $\vv(X_{k,j_k})$ as a substitute. On the other hand, by Lemma~\ref{lem:non-interfering}, such $\vv(X_{k,j_k})$ is non-interfering from $d_i$'s perspective for all $i\in (S(X_{k,j_k})\cup\{k\})$. Therefore, instead of sending a single packet $\vv(X_{k,j_k})$, it is beneficial to {\em linearly combine} the transmission of two packets $\vv(X_{k,j_k})$ and $\vv(X_{l,j_l})$ together, as long as $l\in S(X_{k,j_k})$ and $k\in S(X_{l,j_l})$. 
Since $\vv(X_{k,j_k})$ is non-interfering from $d_l$'s perspective, it is as if $d_l$ directly receives  $\vv(X_{l,j_l})$ without any interference. Similarly, since $\vv(X_{l,j_l})$ is non-interfering from $d_k$'s perspective, it is as if $d_k$ directly receives  $\vv(X_{k,j_k})$ without any interference. By generalizing this idea, the PE scheme first selects a $T\subseteq [K]$ and then constructs a $\vv_\text{tx}$ that can serve all destinations $k\in T$ simultaneously by mixing the corresponding non-interfering vectors.

{\bf Creating new coding opportunities while exploiting the existing coding opportunities:}\quad As discussed in the example of Section~\ref{subsec:example}, the suboptimality of the existing 2-phase approach for $K\geq 3$ destinations is due to the fact that it fails to create new coding opportunities while exploiting old coding opportunities. The PE scheme was designed to solve this problem.
Let us assume that the {\sc Packet Selection} in Line~\ref{line:targeting-set} chooses the $X_{k,j}$ such that $S(X_{k,j})=T\backslash k$. That is, we choose the $X_{k,j}$ that can be mixed with those $(s,d_l)$ sessions with $l\in S(X_{k,j})\cup\{k\}=T$. Then Line~\ref{line:S-update} of the {\sc Update} guarantees that if some other $d_i$, $i\notin T$, overhears the coded transmission, we can update $S(X_{k,j})$ with a strictly larger set $S(X_{k,j})\cup S_\text{rx}$. Therefore, new coding opportunity is created since we can now mix more sessions (all $d_i$, $i\in S(X_{k,j})$) together with $X_{k,j}$. Note that the coding vector $\vv(X_{k,j})$ is also updated accordingly. The new $\vv(X_{k,j})$ represents the necessary ``code alignment" in order to utilize this newly created coding opportunity. The (near-) optimality of the PE scheme is rooted deeply in the concept of code alignment, which aligns the ``non-interfering subspaces" through the joint use of $S(X_{k,j})$ and $\vv(X_{k,j})$.

\subsection{Analysis Of The PE Scheme}

One advantage of a PE scheme is that although different packets $X_{k,j_k}$  and $X_{i,j_i}$ with $k\neq i$ may be mixed together, the corresponding evolution of $X_{k,j_k}$ (the changes of $S(X_{k,j_k})$ and $\vv(X_{k,j_k})$)  are independent from the evolution of $X_{i,j_i}$ (see Line~\ref{line:update-4-K} of the {\sc Update}). Also by Lemma~\ref{lem:non-interfering}, two different packets $X_{k,j_k}$  and $X_{i,j_i}$ can share the same time slot without interfering each other as long as $i\in S(X_{k,j_k})$ and $k\in S(X_{i,j_i})$. As a result, the throughput analysis can be done by focusing on the individual sessions separately, and considering how many time slots from different sessions can be combined together. The achievability results are proven by analyzing the throughput of the PE scheme with carefully designed mechanisms of choosing the set $T$ and the corresponding target packets $\{X_{k,j_k}:k\in T\}$ of a generic PE scheme.

\section{Numerical Evaluation\label{subsec:sim}}

We first notice that both the inner and outer bounds are linear programming (LP) problems and can be evaluated by any LP solvers. We perform numerical evaluation for spatially independent 1-to-$K$ broadcast PECs with COF by randomly varying the values of the marginal success probabilities $(p_1,\cdots, p_K)$.
Note that although there is no tightness guarantee for $K\geq 4$ except in the one-sidedly fair rate region, in all our numerical experiments with $K\leq 6$
(totally $3\times 10^4$ of them), we have not found any instance of the input parameters $(p_1,\cdots, p_K)$, for which the gap between the outer and inner bounds is greater than the numerical precision of the LP solver. This shows that Propositions~\ref{prop:outer} and~\ref{prop:ach2} effectively describe the capacity region from the practical perspective.

\begin{figure}
\centering
\subfigure[\label{fig:sym124}]{\includegraphics[width=4.25cm]{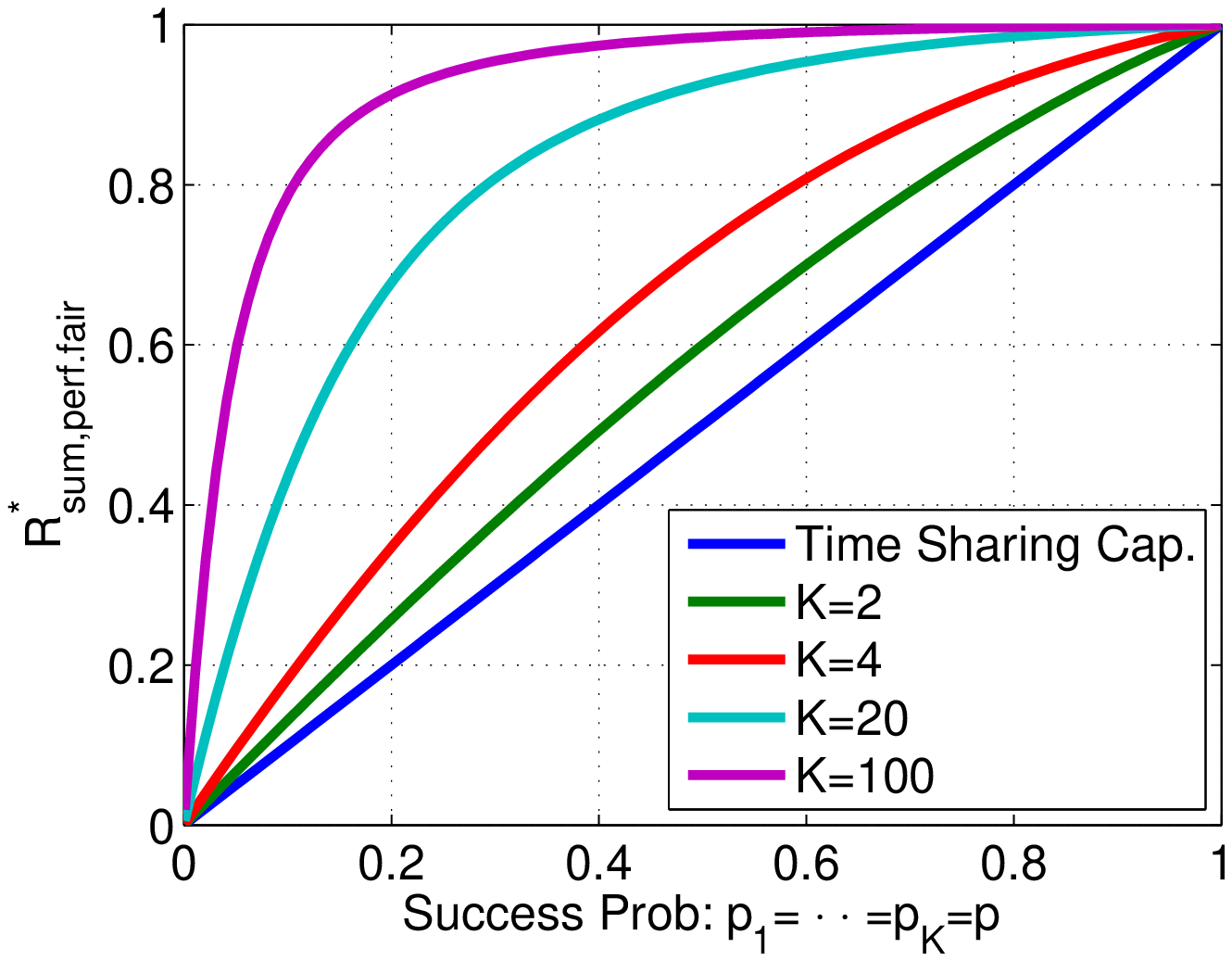}}
\subfigure[\label{fig:asym6}]{\includegraphics[width=4.25cm]{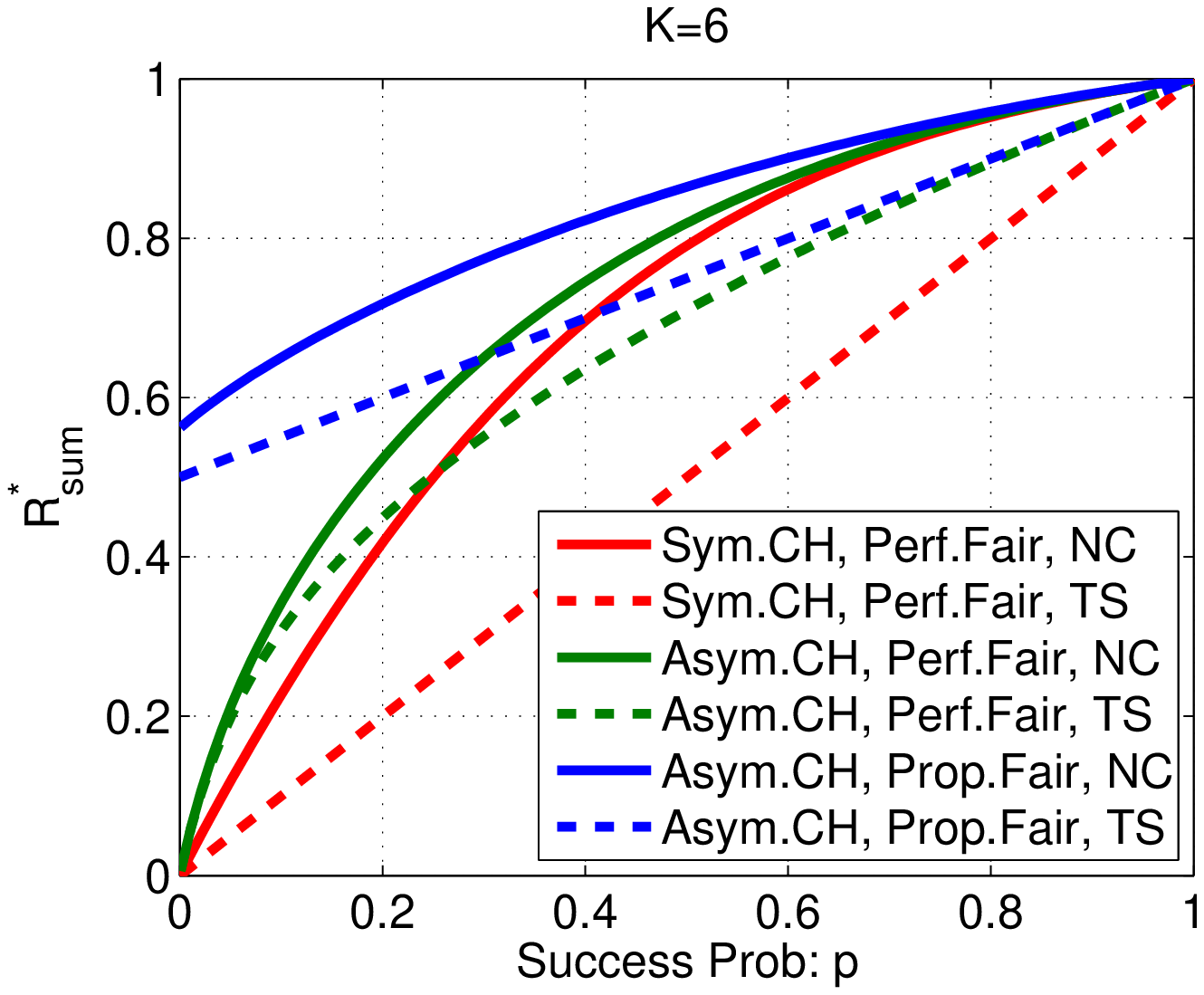}}
\caption{(a) The sum-rate capacity $R^*_{\text{sum,perf.fair}}$ in a perfectly fair system versus the marginal success probability $p$ of a symmetric, spatially independent 1-to-$K$ broadcast PEC, $K=2$, 4, 20, and 100. (b) The sum-rate capacities for a 6-destination heterogenous channel profiles with the success probabilities $p_1$ to $p_6$ evenly spaced between $(p,1)$. }
\end{figure}

To illustrate the network coding gain, we compare the sum-rate capacity versus the sum rate achievable by time sharing. Fig.~\ref{fig:sym124} considers symmetric, spatially independent PECs with marginal success probabilities $p_1=\cdots=p_K=p$. We plot the sum rate capacity $R_{\text{sum,perf.fair}}^*$ versus $p$ for a perfectly fair system $R_1=\cdots=R_K$. 
As seen in Fig.~\ref{fig:sym124}, the network coding gains are substantial when we  have $K\geq 4$ destinations.  It can also be proven that for any $p\in(0,1]$,  $R_{\text{sum,perf.fair}}^*$ approaches one as $K\rightarrow\infty$, which was first observed in \cite{LarssonJohansson06}.

We are also interested in the sum rate capacity under asymmetric channel profiles (also known as heterogeneous channel profiles). Consider asymmetric, spatially independent PECs. For each $p$ value, we let the channel parameters $p_1$ to $p_K$ be equally spaced between $(p,1)$, i.e., $p_k=p+(k-1)\frac{1-p}{K-1}$. We then plot the sum rate capacities for different $p$ values. (In this experiment, the outer and inner bounds in Section~\ref{sec:main} meet for all different $p$ values.) Fig.~\ref{fig:asym6} describes the case for $K=6$. We plot the curves for perfectly fair ($R_1=\cdots =R_K$) and proportionally fair ($R_k\propto p_k$) systems, respectively. For comparison, we also plot the time-sharing capacity under the heterogeneous channel profile. For comparison between symmetric (homogeneous) and asymmetric (heterogeneous) channel profiles, we plot the sum-rate capacity for symmetric channels as well. As can be seen in Fig.~\ref{fig:asym6}, network coding again provides substantial improvement for all $p$ values. However, the gain is not as large as in the case of symmetric channels. The results show that for practical implementation, it is better to group together all the sessions of similar marginal success rates and perform intersession network coding within the same group.

\section{Conclusion\label{sec:conclusion}}

In this work, we have proposed a new class of intersession network coding schemes, termed the packet evolution (PE) schemes, for the broadcast PECs with COF. Based on the PE schemes, we have
derived the capacity region for general 1-to-3 broadcast PECs, and a pair of capacity outer and inner bounds for general 1-to-$K$ broadcast PECs, both of which can be easily evaluated by any linear programming solver for the cases $K\leq 6$. It has also been proven that the outer and inner bounds meet for two classes of 1-to-$K$ broadcast PECs: the symmetric broadcast PECs, and the spatially independent broadcast PECs with the one-sided fairness rate constraints. Extensive numerical experiments have shown that the outer and inner bounds meet for almost all broadcast PECs encountered in practical scenarios.


\section*{Acknowledgment}

This work was supported in parts by NSF grants CCF-0845968 and CNS-0905331. The author would also like to thank  Profs.\ Anant Sahai and David Tse for their suggestions.

\bibliography{ntwkcoding,nc_sys,chihw}
\bibliographystyle{IEEEtranS}

\end{document}